\documentclass[runningheads]{llncs}

\usepackage{times}
\usepackage{soul}
\usepackage{url}
\usepackage[hidelinks]{hyperref}
\usepackage[utf8]{inputenc}
\usepackage[small]{caption}
\usepackage{graphicx}
\usepackage{mathtools}
\usepackage{amsmath}
\usepackage{amsfonts}
\usepackage{booktabs}
\usepackage{algorithm}
\usepackage{algorithmic}
\usepackage{mathabx}
\usepackage{comment}
\usepackage{soul,color}
\usepackage{changebar}

\newcommand{\VA}{\mathsf{VA}}

\newcommand{\all}[1]{{[\![ #1 ]\!]}}

\newcommand{\CGS}{\texttt{vCGS}}
\newcommand{\CGSname}{\texttt{M}}

\renewcommand{\top}{\mathrm{tt}}
\renewcommand{\bot}{\mathrm{ff}}

\newcommand{\new}[1]{{\color{red}{#1}}}

\newcommand{\Ag}{{Ag}}

\newcommand{\tval}{\mathsf{tv}}

\newcommand{\naww}[1]{{\langle\!\langle #1 \rangle\!\rangle}}

\newcommand{\guard}{\mathsf{g}}

\def\sd{\longrightarrow}

\def\qed{\hfill{\qedboxempty}      
	\ifdim\lastskip<\medskipamount \removelastskip\penalty55\medskip\fi}

\def\qedboxempty{\vbox{\hrule\hbox{\vrule\kern3pt
			\vbox{\kern3pt\kern3pt}\kern3pt\vrule}\hrule}}

\def\restr#1{\,\mbox{\rule[-4pt]{0.5pt}{13pt}}_{#1}}

\title{Model Checking Strategic Abilities in Information-sharing Systems}

\author{
	F.~Belardinelli, I.~Boureanu, C.~Dima, and V.~Malvone}

\institute{}

\begin{document}

\maketitle
\begin{abstract}
 
~We introduce a subclass of concurrent game structures (CGS) with imperfect information 
 in which agents are endowed with private data-sharing capabilities. Importantly, our CGSs are such that it is still decidable to model-check these CGSs against a relevant fragment of ATL.
These systems can be thought as a  generalisation of architectures allowing information forks, 
in the sense that, in the initial states of the system, we allow information forks from agents outside a given set $A$ to agents inside this  $A$. 
For this reason, together with the fact that the communication in our models underpins a specialised form of broadcast,  we call 
our formalism  \emph{$A$-cast systems}.

~To underline, the fragment of ATL for which we show the model-checking problem to be decidable  over $A$-cast is a large and significant one; it expresses coalitions over agents in any subset of the set $A$. Indeed, as we show,  our systems and this ATL fragments can encode security problems that are notoriously hard to express faithfully: terrorist-fraud attacks in identity schemes.


\end{abstract}

\section{Introduction}

\subsection*{The Increasing Importance of Private-Information Sharing \& Collusion-encoding.} 

Sharing information over private channels is a key issue in many ICT
  applications, from robotics to information
  systems~\cite{Samaila2017}. Safeguarding privacy is also
  of utmost concern, also in the context of EU's newly-enforced General
  Data Protection Regulation (GDPR)~\cite{eu:gdpr}.  So, being able to
  explicitly
  model data-sharing over private channels in
  multi-agent systems (MAS) or concurrent systems is crucial.

  On the other hand, there are numerous ICT applications, such as identity-schemes and
 distributed ledgers, where the threat of adversaries
 colluding \emph{privately} with inside accomplices is greater than
 that of classical, outsider-type attackers.  
 That said, verifying security against collusions-based attacks
such as terrorist-frauds~\cite{Bengio91} in identification schemes is
a notoriously difficult problem in security
verification~\cite{db}. And, even the most recent formalisms
attempting this~\cite{AD:2019} fall short of encapsulating the
strategic nature of this property or of the threat-model it views;
instead, \cite{AD:2019} looks at correspondence of events in applied
$\pi$-calculus (i.e., if agents did a series of actions with one
result, then later another fact must be the case).  
Instead, they use arguably crude approximations
 (based on causalities/correspondence of events in process
 calculi~\cite{Blanchet01}) of what a collusive attack  would
 be. But, indeed, such attacks resting on not one by two
 interdependent strategies are notoriously hard to
 model~\cite{db}. Meanwhile,
 expressive logics for strategic reasoning have been
 shown effective in capturing intricate security requirements, such as
 coercion-resistance in
 e-voting~\cite{TabatabaeiJR16,BelardinelliCDJ17,Selene18}.
 Thus, the further study of strategic abilities under collusion is of
 timely interest.

Looking into this,  we observe that the typical formalisms for encoding
multi-agent systems (MAS), such as interpreted systems \cite{FHMV95},
support the analysis of strategy-oriented properties.   
   Yet, these systems are generally
inspired by concurrent action models and, so, explicit expression
therein of data-sharing between agents is difficult to model
naturally.  Other frameworks, such as
dynamic epistemic logic (DEL)~\cite{hvdetal.del:2007,Maffre16} or
propositional assignments~\cite{CLPC2015}, provide agents with more
flexibility on data sharing, but do not allow for the full power of
strategic reasoning.
To fill this gap, the formalism of \emph{\CGS{}} was recently proposed in \cite{BBDM19}. \CGS{} supports the
 explicit expression of private-data sharing in MAS.
 That is,  a \CGS{} encodes syntactically and in a natural
 semantics a ``MAS with 1-to-1 private-channels'': agent $a$ and agent
 $b$ have an explicit syntactic/semantic endowment to ``see'' some of
 each others' variables, without other agents partaking in this. However, unfortunately, verifying strategy-based logic specification on \CGS{} has been shown undecidable. 
 So, it would interesting and timely, to see if we can maintain the needed feature of private-information sharing that \CGS{} have and recover the decidability of model checking logics that express  strategic abilities.
 
 \vspace*{0.3cm}
 
 To sum up, in this paper, \emph{we look at tractable verification of logics of strategic-abilities in systems capable of expressing private data-sharing between processes; we show an important decidability result in this space. We do this with the view of applying this to ICT systems where the analysis of strategic behaviour is crucial: e.g., collusion-based attacks in security settings.}


\section{Context  \& Contribution}

\textbf{Context on Alternating-Time Temporal Logic.} Alternating-time temporal logic (ATL) \cite{AlurHenzingerKupferman02} is a powerful extension of 
branching-time logics which indeed provides a  framework for reasoning about strategic behaviours 
in concurrent and multi-agent systems.
From the very beginning, ATL was proposed also for agents with imperfect information.
%
In ATL semantics, agents can have memoryless strategies or recall-based/memoryful strategies; 
in the latter, they ``remember'' all their history of moves, whereas in the former  -- they recall none. 
Also,  ATL with imperfect information comes with several semantic flavours~\cite{Jamroga-Hoek} 
(subjective, objective or common-knowledge interpretation), created by  
different views on the intricate relationships between its temporal and epistemic aspects.

\textbf{Context on ATL Decidability.}  There have been several proposals for decidable fragments of ATL with imperfect information,
which arose via two alternative  avenues: 
(a) by imposing a memoryless semantics for the ATL operators~\cite{LomuscioQuRaimondi15}, approach implemented in the MCMAS tool; 
(b) by making structural restrictions on the indistinguishability relations inside the game structures, 
so that one looks at ATL just with distributed knowledge \cite{DEG10}, or ATL with hierarchical knowledge \cite{BerthonMM17} 
or ATL over  broadcasting systems \cite{BLMR17,BelardinelliLMR17b}.
Some other decidable fragments can be further obtained by adapting  
decidable cases of the distributed-synthesis problem, or the 
the existence of winning strategies in multi-agent games with imperfect information.  
We therefore may generalise the decidability of the distributed-synthesis problem in architectures without information forks~\cite{FinkbeinerScheweLICS05},
or utilise the decidability of the problem of the existence of a winning strategy in 
multi-player games with finite knowledge gaps~\cite{berwanger-knowledge-gaps}. 


\subsection*{Contribution}

\textbf{I. Our Private Data-Sharing Systems.}
In this paper, we propose a new class of game structures is called \emph{$A$-cast game structures}, where $A$ is a set of processes/agents.
In $A$-cast game structures, when some outsider (agent not in $A$) sends some information to some agent in $A$, 
the same information is seen to all agents in the coalition $A$.

$A$-cast game structures are introduced using the formalism called {\CGS{}}  recently proposed\footnote{\CGS{} can be seen as a generalisation of Reactive Modules Games with imperfect information \cite{GutierrezPW16}, in that 
each agent may dynamically modify the visibility of the atoms she ``owns'', by either disclosing or hiding the value of that atom to some other agent.} in \cite{BBDM19}. 
We can introduce $A$-cast systems directly as game structures, yet we chose to use \CGS{} simply as it
allows for an elegant formalisation of the information-flow between outsiders to the coalition $A$. So, as a consequence of this presentation choice, $A$-cast game structures can be seen as a specialisation of the {\CGS{}} in \cite{BBDM19}.

$A$-cast game structures are in fact a strict generalisation of both architectures without information forks \cite{FinkbeinerScheweLICS05} 
and broadcast systems \cite{vanderMeyden2005}.
On the one hand, the extension of \cite{FinkbeinerScheweLICS05} comes from the fact that, in $A$-cast game structures,
the set of initial states is arbitrary, as well as the epistemic relations defining the information available to agents in initial states.
On the other, modelling these features using the setting of distributed architectures in \cite{FinkbeinerScheweLICS05}, where a single initial state is imposed, 
requires the Environment agent to send the initial information to each agent by some private channel, hence creates an information fork.

\textbf{II. Our Decidable Fragment of ATL.}
We now describe the class of formulas for which the  model-checking problem is  decidable on our  $A$-cast game structures, as well as some details linked to this. 
This ATL fragment  is composed of formulas which utilise only 
coalition operators involving a set of agents $B\subseteq A$, where $A$ is the set of agents/processes describing the $A$-cast system.

To obtain this result, the key point is that two action-histories starting in the same initial state, both being
generated by the same joint strategy for coalition $A$ and 
being in the same common-knowledge indistinguishability class for $A$ are in fact in the same distributed-knowledge indistinguishability class for $A$.
This property allows for the  design of a finitary information-set construction for the appropriate 
multi-player game with imperfect information \cite{berwanger-tracking},
needed to decide each coalition operator.

\textbf{III. Our Case Study and Its Need of Verification against Formulae with Nested-ATL Modalities.}
 We  provide a case study which shows the relevance of our new decidable class of ATL model-checking.
This case study lives in the cyber-security domain: the identity-schemes and  threats of collusion-based attacks therein (i.e., terrorist-fraud attack).
Concretely, we model the  distance-bounding (DB)  protocol by  Hancke and Kuhn~\cite{HK05}
 as an $A$-cast \CGS, and the existence of a terrorist-fraud attack on this protocol
 as an $ATL$ formula involving a non-singleton coalition operator. In fact, we also note that the $ATL$ formula which specifies the existence of a terrorist-fraud attack in our case study 
is the first type of formula requiring nesting of $ATL$ operators.

Hence, the model-checking algorithm proposed in this paper can be applied to this case study,
while other existing decidable model-checking or distributed-synthesis frameworks can treat it (due to the formula it requires). 
The only  exception would be the utilisation of a memoryless semantics of $ATL$ \cite{LomuscioQuRaimondi15},
which would be applicable since our case study is in fact a model in which  there exists a finite number of 
lasso-type infinite runs. Yet,  specifying our case-study in a formalism like that of~\cite{LomuscioQuRaimondi15}  would require an explicit encoding/``massaging'' of the agent memory
into the model. 
In other words, our algorithm synthesises memoryful strategies and hence
allows working with a given model without explicit encoding of agent memory into the model.


\vspace*{0.3cm}

{\bf Structure.}
In Section~\ref{background}, we present general
preliminaries on Alternating-time Temporal Logic ($ATL$) and
concurrent game structures with imperfect information
(iCGS) \cite{AlurHenzingerKupferman02}.  In Section~\ref{visAg}, {we
recall MAS with private-channels, called \CGS\ \cite{BBDM19}.  
On \CGS{}, we provide ATL with our main formalism called \emph{$A$-cast systems}, that is a semantics
under the assumptions of imperfect information and perfect recall,
under subjective interpretation \cite{Jamroga-Hoek}.  In
Section~\ref{strcom}, we present the main decidability result of model checking ATL on top of $A$-cast systems.  
In Section~\ref{example}, we show how use the logic
and result here to check he existence of collusion-based attacks in
secure systems.  
Section \ref{sec:relwork} discusses related work and
future research.




\section{Background on Alternating-time Temporal Logic}
\label{background}

We here recall background notions on concurrent game structures with
  imperfect information (iCGS) and Alternating-time Temporal Logic
  (ATL) \cite{AlurHenzingerKupferman02}.
%
We denote the length of a tuple $r$ as $|r|$, and its $i$th
element either as $r_i$ or $r[i]$.
For $i \leq |r|$, let $r_{\geq i}$ or $r[\geq i]$ be the suffix
$r_{i},\ldots, r_{|r|}$ of $r$ starting at $r_i$ and let $r_{\leq i}$
or $r[\leq i]$ be the prefix $r_{1},\ldots, r_{i}$ of $r$ ending at
$r_i$.  Further, we denote with $last(r)$ the last element $r_{|r|}$
of $r$.  Hereafter, we assume a finite set $Ag = \{1, \ldots, m\}$ of
{\em agents} and an infinite set $AP = \{p_1, p_2, \ldots \}$ of {\em
atomic propositions} (atoms).
\begin{definition}[iCGS]\label{icgs}
Given sets $Ag$ of agents and $AP$ of atoms, a {\em CGS with
  imperfect information} is a tuple $\CGSname
 = \langle  S,  S_0, \{
  Act_a \}_{a \in Ag}, \allowbreak  \{\sim_a \}_{a \in Ag}, P, \tau,  \pi
\rangle$ where:
\begin{itemize}
\item
$S$ is the set of {\em states}, with $S_0 \subseteq S$ the set of {\em initial states}.

\item
For every agent $a \in Ag$, $Act_a$ is the set of {\em actions}
  for $a$. Let $Act = \bigcup_{a \in \Ag} Act_a$ be the set of all
  actions, and $ACT = \prod_{a \in \Ag} Act_a$ the set of all {\em joint
  actions}.


\item
For every agent $a \in Ag$, the {\em indistinguishability
	relation} $\sim_a$ is an equivalence relation on $S$.
	\item 

		  For every $s \in S$ and $a \in Ag$, the {\em
          protocol function} $P: S \times Ag \to
          (2^{Act} \setminus \emptyset)$ returns the non-empty set
          $P(s,a) \subseteq Act_a$ of actions enabled at $s$ for $a$
          s.t.~$s \sim_a s'$ implies $P(s,a) = P(s',a)$.


	\item
$\tau : S \times ACT \to S$ is the (partial) {\em
          transition function} s.t.~$\tau(s, \vec{\alpha})$ is
          defined iff $\alpha_a \in P(s,a)$ for every $a \in Ag$.


	\item
$\pi : S \to 2^{AP}$ is the {\em labelling function}.
\end{itemize}
\end{definition}

An iCGS describes the interactions of set $Ag$ of
agents.
Every agent $a$ has {\em imperfect information} about the global state
  of the iCGS, as in every state $s$ she considers any state
  $s' \sim_a s$ as (epistemically) possible \cite{FHMV95,Jamroga-Hoek}.

To reason about the strategic abilities of agents in iCGS,
we adopt the Alternating-time Temporal Logic $ATL$.
%


\begin{definition} \label{def:ATL*}
	Formulas in $ATL$ are defined as
	follows, for $q \in AP$ and $A \subseteq \Ag$.
\[
		\varphi ::=  q \mid \neg \varphi  \mid \varphi \land \varphi \mid \naww{A} X \varphi \mid \naww{A} \varphi U \varphi \mid \naww{A} \varphi R \varphi 
\]
\end{definition}
The strategy operator $\naww{A}$ is read as ``coalition $A$
can achieve \ldots''. 
%

The \emph{subjective} interpretation of $ATL$ \emph{with imperfect information and perfect recall} \cite{Jamroga-Hoek} is defined
on iCGS as follows.  Given an iCGS $\CGSname$, a {\em path} $p$ is a (finite or
infinite) sequence $s_1 s_2\dots$ of states such that for every
$i \geq 1$ there exists some joint action $\vec{\alpha} \in ACT$ such
that $s_{i+1} = \tau(s_i, \vec{\alpha})$. A finite, non-empty path
$h \in S^+$
is called a {\em history}.
Hereafter, we extend the indistinguishability relation $\sim_a$
to histories:
$h \sim_a h'$ iff $|h| = |h'|$ and for every $i \leq
|h|$, $h_i \sim_a h'_i$.
\begin{definition}[Strategy] \label{uniformity}
	A {\em uniform, memoryful strategy } for agent $a \in \Ag$ is a
        function $f_a : S^+ \to Act_a$ such that for all histories $h, h' \in
        S^+$, (i) $f_a(h) \in P(last(h),a)$; and (ii) if $h \sim_a h'$ then $f_a(h) = f_a(h')$.
\end{definition}

Given a {\em joint strategy} $F_A = \{ f_a \mid a \in A \}$ for
coalition $A \subseteq Ag$, and history $h \in S^+$, let $out(h, F_A)$
be the set of all infinite paths $p$ whose initial segment is
indistinguishable from $h$ and consistent with $F_A$, that is,
$out(h, F_A) = \{ p \mid p_{\leq|h|} \sim_a h \text{ for some } a \in
A,
\text{ and for all } i \geq |h|, p_{i+1} = \tau(p_{i}, \vec{\alpha}),
\text{ where for all } a \in A, \alpha_a = f_a(p_{\leq i})\}$.

We now assign a meaning to $ATL$ formulas on iCGS.
\begin{definition}[Satisfaction] \label{satisfaction}
	The satisfaction relation $\models$ for an iCGS $\CGSname$, path $p$,
        index $i \in \mathbb{N}$, and $ATL$ formula $\phi$ is
        defined as follows (clauses for Boolean operators are immediate and thus omitted):
	
	{\small
	\begin{tabbing}
		$(\CGSname, p, i) \models q$ \ \ \ \ \ \ \ \ \ \ \ \ \ \ \ \ \ \ \  \=  iff   \ \= $q \in \pi(p_i)$\\
	  $(\CGSname, p,i) \models \naww{{A}} X \varphi$  \>iff\> for some strategy $F_A$, for all $p' \in out(p_{\leq i}, F_A)$, $(\CGSname, p',i+1) \models \varphi$\\
	  $(\CGSname, p,i) \models \naww{{A}} \varphi_1 U \varphi_2$  \>iff\> for some strategy $F_A$, for some $p' \in out(p_{\leq i}, F_A)$,  \\
	  \>\>$\exists j\geq i$ s.t.~$(\CGSname, p',j) \models \varphi_2$ and $\forall i\leq k<j$, $(\CGSname, p',j) \models \varphi_1$\\
	  $(\CGSname, p,i) \models \naww{{A}} \varphi_1 R \varphi_2$  \>iff\> for some strategy $F_A$, for all paths $p' \in out(p_{\leq i}, F_A)$, \\
\> \> $\forall j\geq i, either (\CGSname, p',j) \models \varphi_2$, or $\exists i \leq k \leq j$ s.t.~$(\CGSname, p', k) \not \models \varphi_1$
	\end{tabbing}
}

\end{definition}
Operators $\all{A}$,
`eventually' $\naww{A} F$, and `globally' $\naww{A} G$ can be introduced as usual.

A formula $\varphi$ is {\em true at state $s$}, or $(\CGSname, s) \models
\varphi$, iff for all paths $p$ starting in $s$,
$(\CGSname, p,1) \models \varphi$.
A formula $\varphi$ is {\em true in an iCGS $\CGSname$}, or $\CGSname \models
\varphi$, iff for all initial states $s_0 \in S_0$, $(\CGSname, s_0) \models \varphi$.



Our choice for the subjective interpretation of ATL is motivated by the fact that it
allows us to talk about the strategic abilities of agents as depending
on their knowledge, which is essential in the analysis of
security protocols.
We illustrate this point in Section~\ref{example}.

%

Hereafter we tackle the following major decision problem.
%
\begin{definition}[Model checking problem] 
Given a iCGS $\CGSname$ and an $ATL$
formula $\varphi$, the model checking problem amounts to determine
whether $\CGSname \models \varphi$.
\end{definition}



\section{Agents with Visibility-Control} \label{visAg}

We now provide details on a notion of agent presented in  \cite{BBDM19}. 
Such an agents can change the truth-value of the atoms she controls, and make atoms visible
to other agents.
\begin{definition}[Visibility Atom] \label{vis_atom}
Given atom $v \in AP$ and agent $a \in Ag$,
$vis(v,a)$ denotes a \emph{visibility atom} expressing intuitively
that the truth  value of $v$ is visible to $a$.
By \emph{$\VA$} we denote the set of all visibility atoms
$vis(v, a)$, for $v \in AP$ and $a \in Ag$.  By \emph{$\VA_a$}=$\{
vis(v, a) \in \VA \mid v \in AP \}$ we denote the set of visibility
atoms for agent $a$.
\end{definition}

Importantly, the notion of visibility in Def.~\ref{vis_atom} is
dynamic, rather than static, as it can change \emph{at run time}.
That is,
agent $a$,
can make atom $v$ visible (resp.~invisible) to agent $b$ by setting
the truth value of atom $vis(v,b)$ to true (resp., false).
%
%
This intuition motivates the following definition.
\begin{definition}[Visibility-Controlling Agent: Syntax]
\label{vis_agent}
  Given set $AP$ of atoms, an {\em agent}
is a tuple $a = \langle V_a, 
\texttt{update}_a \rangle$ such that
\begin{itemize}
  \item
%
%
$V_a \subseteq AP$ is the set of atoms controlled by agent $a$;


\item 
$ \texttt{update}_a$ is a finite set of {\em update commands}
of the form:
\begin{eqnarray*}
	\gamma   ::=  &  \varphi \rightsquigarrow & v_{1} := \tval, \ldots, v_{k} := \tval, vis(v_{k+1}, a_1) := \tval \ldots, vis(v_{k+m}, a_m) := \tval
\end{eqnarray*}
%
where each $v_i \in V_a$ is an atom controlled by $a$ that occurs at most
once, {\em guard} $\varphi$ is a Boolean formula over $AP \cup \VA_a$,
all $a_1, \ldots, a_m$ are agents in $Ag$ different from $a$,
and $\tval$ is a truth value in  $\{\top, \bot \}$.
\end{itemize}
We denote with $\guard(\gamma)$ and $asg(\gamma)$
the guard (here, $\varphi$)  and assignment of $\gamma$, respectively.


%
\end{definition}


The intuitive reading of a guarded command $\varphi \rightsquigarrow asg$ is: if guard $\varphi$ is
true, then the agent executes
assignment $asg$.
By Def.~\ref{vis_agent} every agent $a$ controls the truth value of every
atom $v \in V_a \subseteq AP$.
%
Moreover, differently from \cite{AlurH99b,HoekLomuscioWooldridge06},
she can switch the visibility for some other agent $b$ over her own
atoms, by means of assignments $vis(v, b) := \tval$.
%
Lastly,
by requiring that all $a_1, \ldots, a_m$ are
different from $a$, we guarantee that agent $a$ cannot loose
visibility of her own atoms.
Hereafter, we assume that control is {\em exclusive}: for any two
distinct agents $a$ and $b$, $V_a \cap V_b = \emptyset$, i.e., the
sets of controlled atoms are disjoint.
Then, we
often talk about the {\em owner} $own(v) \in Ag$ of an atom $v \in
AP$.

The agents in Def.~\ref{vis_agent} can be thought of as a compact,
program-like specifications of a multi-agent system.  The actual
transitions in the system are described by a specific class of concurrent game
structures with imperfect information,
%
that we call {\em iCGS with visibility-control (\CGS)}.  In
  particular, a state in a \CGS\ is a set $s \subseteq AP \cup \VA$ of
  propositional and visibility atoms such that for every $a \in Ag$,
  $v \in V_a$, $vis(v, a) \in s$, that is, every agent can see the
  atoms she controls.  Then, given state $s$ and $a \in Ag$, we
  define \emph{$Vis(s,a) = \{ v \in AP \mid vis(v,a) \in s \}$} as the
  set of atoms visible to agent $a$ in $s$. Notice that, by definition,
  $V_a \subseteq Vis(s,a)$ for every state $s$.

%
%


\begin{definition}[{Visibility-Controlling Agents: Semantics}] \label{vis_agent_sem}
Given a set $Ag$ of agents as in Def.~\ref{vis_agent} over a set $AP$ of atoms, and a set $S_0$ of initial states,
an \emph{iCGS with visibility-control (\CGS)} is an iCGS $\CGSname
  = \langle S, S_0, \{Act_a \}_{a \in Ag}, \{\sim_a \}_{a \in Ag},
  P, \tau, \pi \rangle$ where:
\begin{itemize}
	\item $S$  and  $S_0$ are the set of all states and the set of  initial states, respectively;

	\item
For every agent $a \in Ag$, $Act_a = \texttt{update}_a$.


      \item
The {\em indistinguishability relation} is defined so
        that for every $s, s' \in S$, $s \sim_a s'$ iff $Vis(s,a) =
        Vis(s',a)$ and for every $v \in Vis(s,a) = Vis(s',a)$, $v \in
        s$ iff $v \in s'$.
%





          	\item
          For every $s \in S$ and $a \in Ag$, the
          {\em protocol function} $P: S \times Ag \to 2^{Act}$ returns set $P(s,a)$ of \texttt{update}-commands
$\gamma \in Act_a$ s.t.~$AP(\guard(\gamma)) \subseteq Vis(s,a)$ and $s
          \models \guard(\gamma)$, where $AP(\phi)$ is the set of
          atoms occurring in $\phi$.
That is, all atoms appearing
in the guard are visible to the agent
and the guard is true.
We can enforce $P(s,a) \neq \emptyset$ for all $s \in S$ and $a \in
Ag$, by introducing a {\em null} command
$\top \rightsquigarrow \epsilon$, with trivial guard and empty
assignment.
        We can check that if $s \sim_a s'$ then $P(s,a) =
        P(s',a)$.


\item
The {\em transition function} $\tau : S  \times  ACT  \to  S$ 
is s.t. $\tau(s, (\gamma_1, \ldots, \allowbreak \gamma_m))  =  s'$ holds iff:
	\begin{itemize}
		\item For every $a \in Ag$, $\gamma_a \in P(s,a)$.
		\item For every $v\in AP$ and $own(v) \in \Ag$,
$v \in s'$ iff either $asg(\gamma_{own(v)})$ contains an assignment $v := \top$ or, 
		$v \in s$ and $asg(\gamma_{own(v)})$ does not contain an assignment $v := \bot$.

		\item For every $v\in AP$ and $own(v) \in \Ag$,
$vis(v,a) \in s'$ iff either $asg(\gamma_{own(v)})$ contains an assignment $vis(v,a)  :=  \top$ or, 
		$vis(v,a) \in s$ and $asg(\gamma_{own(v)})$ does not contain an assignment $vis(v,a)  :=  \bot$.
\end{itemize}


         \item
The labelling function $\pi : S \to 2^{AP}$ is the identity function restricted
to $AP$, that is, every state is labelled with the propositional atoms
belonging to it.

\end{itemize}
\end{definition}



%



Definitions \ref{vis_agent} and \ref{vis_agent_sem} implicitly state
that agents in a \CGS\ interact by means of strategies, defined, as
for iCGS, as suitable functions $f_a : S^+ \rightarrow Act_a$.  This
means that the set of guarded commands in an agent specification
describe all possible interactions, from which agents choose a
particular one. The chosen command can be seen as a refinement of the
corresponding guarded command, in the form of a (possibly infinite)
set of guarded commands $(\gamma_h)_{h\in S^+}$ where each $\gamma_h$
is obtained from the unique command $\gamma = f_a(h)$ by appending to
the guard $\guard(\gamma)$ of $\gamma$ some additional formula
uniquely identifying the observations of history $h$ by agent $a$.

Clearly, $\CGS$ are a subclass of iCGS, but still expressive enough to
capture the whole of iCGS.
%
Indeed, \cite{BBDM19b} provides a polynomial-time reduction of the
model checking problem for $ATL$ from
 iCGS
 to \CGS.
 Intuitively, all components of an iCGS can be
encoded by using propositional and visibility atoms only.
In particular, since model checking $ATL$ under perfect information
and perfect recall is undecidable on iCGS \cite{DimaT11}, it follows
that this is also the case for
\CGS.
This also means that, in general, the "implementation" of a strategy
accomplishing some ATL objective along the above idea of "refining"
the guarded commands may yield an infinite set of guarded commands
$(\gamma_h)_{h\in S^+}$.  The possibility to obtain a finite, refined
set is therefore intimately related with the identification of a
subclass of
\CGS\ in which the model checking problem for ATL is decidable. We address this problem in the next section
and provide a relevant case study for the relevant subclass in Section~\ref{example}.

\section{$\CGS$ with Coalition Broadcast} \label{strcom}
\label{decproof}

Given the undecidability results
 in~\cite{BBDM19}, it is of interest to find restrictions on the class
 of \CGS\ or language $ATL$ that allow for a decidable model checking
 problem.
Hereafter, we introduce one such restriction that
is also relevant for the security scenario presented in
Section~\ref{example}; thus allowing us to verify it against
collusion.
In that follows we fix a coalition $A \subseteq Ag$ of agents.

\begin{definition}[$A$-cast \CGS] \label{restr}
A \CGS\ {\em with broadcast towards coalition $A$} (or {\em $A$-cast})
satisfies the following conditions for every $a, b \in A$, $c \in Ag\setminus A$, $v \in V_c$ and
$\gamma_c \in Act_c$:

\begin{enumerate} 
\item [{\rm (}Cmp{\rm)}]
$\gamma_c$ contains an assignment of the type $vis(v,a) ::= t$ for some $t \in \{\top,\bot\}$.
\item [\ensuremath{(\dagger^{A})}]
For $t \in \{\top, \bot \}$,
$asg(\gamma_c)$ contains  $vis(v, a)::= t$  iff  it contains $vis(v, b)::= t$ also.
\end{enumerate}

\end{definition}


\begin{remark}
Condition ({\em Cmp}) ensures that each command for an agent
$c$ outside the coalition determines explicitely the visibility of each
variable in $V_c$, which cannot just be copied via some transition.  
Hence, both conditions ({\em Cmp}) and \ensuremath{(\dagger^{A})}, ensure
that, in all states along a run \emph{excepting the initial state} all
agents inside coalition $A$ \emph{observe the same atoms owned by
adversaries in $Ag \setminus A$}. This statement, formalized in Lemma 1 below,
is a generalization of the ``absence of information forks''
from \cite{FinkbeinerScheweLICS05}, which disallows
private one-way communication channels from agents outside the
coalition to agents inside the coalition, in order to avoid
undecidability of their synthesis problem.

The restriction in ({\em Cmp}) however does not apply at initial
states: the indistinguishability relation is arbitrary on initial
states.  This makes systems satisfying Def.~\ref{restr} strictly more
general than systems not having ``information forks''
in \cite{FinkbeinerScheweLICS05}.  To see that, we note that the
notion of strategy in \cite{FinkbeinerScheweLICS05} amounts to the
following: a strategy for agent $a$ is a function $f_a : S^* \to
Act_a$ such that for all histories $h, h' \in S^+$, (i) $f_a(h) \in
P(last(h),a)$; and (ii) if $h \sim_a h'$ then $f_a(h) = f_a(h')$.  The
difference between this definition and ours (where $f_a$ is a function
from $S^+$ to $Act_a$) implies that \emph{there is a unique initial
state}.  This is not the case in our setting.

Indeed, CGS with multiple initial states and an arbitrary
indistinguishability relation on these may be simulated by CGS with a
unique initial state by allowing the \emph{environment}
(or \emph{adversaries}) to do an initial transition to any of the
initial states.  But this additional initial transition must set the
indistinguishability of the resulting state, which means that if we
had two initial states $q_0^1,q_0^2$ in the original CGS with
$q_0^1 \sim_a q_0^2$ but $q_0^1\not\sim_b q_0^2$, this creates an
information fork in the sense of \cite{FinkbeinerScheweLICS05}.  So
our setting allows "information forks only at initial states",
contrary to \cite{FinkbeinerScheweLICS05}.

Note that Def.~\ref{restr} still allows an agent $a\in A$ to have a "private
one-way communication channel" to some agent $b\in Ag$, in the sense
that, for some $v \in V_a$ and reachable state $s \in S$, $v \in
Vis(s,b)$ but $v \not\in Vis(s,c)$ for any other $c \in Ag$, so agent
$b$ sees the value of $v$ which any other agent does not.
\end{remark}

\begin{remark}
\CGS\ constrainted by Def. \ref{restr} are more general than "broadcast" game structures from \cite{BLMR17} or \cite{vanderMeyden2005}.
To see this, note that the setting from \cite{BLMR17} amounts to the
following restriction: for any state $s$ and distinct \emph{joint
actions} $\gamma_1$, $\gamma_2 \in ACT$, if $t_1 = \tau(q,\gamma_1)$
and $t_2 = \tau(q,\gamma_2)$ then $t_1 \not \sim_b t_2$ for any $b \in
Ag$. But this restriction disallows actions by an agent $a \in A$
which update some variable $v \in V_a$ which is invisible to any other
agent, a fact which is permitted by Def. \ref{restr}.
Also our setting is more general than the ``broadcast'' systems
in \cite{vanderMeyden2005}.  In such broadcast systems,
condition \ensuremath{(\dagger^{A})} is to be satisfied by all agents,
both inside and outside the coalition, which is not the case here.
\end{remark}

In what follows,
the notation $\sim_A^E = \bigcup_{a \in A} \sim_a$ stands for the
\emph{group knowledge relation} for coalition $A$, 
$\sim_A^C = \big(\sim_A^E\big)^*$ for the 
{\em common knowledge} for coalition $A$, whereas $\sim_A^D =
\bigcap_{a \in A} \sim_a$ for the {\em distributed knowledge} for $A$ \cite{FHMV95}. Both
notations are overloaded over states and histories: let $\sim$ be any
of $\sim_a$, $\sim_A^C$, or $\sim_A^D$, for $a \in Ag$ and $A
\subseteq Ag$, then for all histories $h$, $h'$, we set $h \sim h'$
iff (i) $|h| = |h'|$, and (ii) for every $i \leq |h|$, $h_i \sim
h'_i$.

The following lemma is essential in proving the decidability result in
Theorem \ref{thm-A-cast}.
Intuitively, it states that histories that are related through the
common knowledge of coalition $A$, are compatible with the same
strategy, and \emph{share the same initial state}, are in fact in the
distributed knowledge of that coalition.  
This can be seen as a generalization of Lemma 2 from \cite{vanderMeyden2005}.

\begin{lemma}
	\label{lemma-two-histories} Let $\CGSname$ be an
        $A$-cast \CGS.  Assume a joint strategy $\sigma_A$ on
        $\CGSname$.  Consider histories $h_1,h_2$ such that
        $h_1 \sim_A^C h_2$,
	$h_1[1] = h_2[1]$, and  $h_i \in out(h_i[1],\sigma_A)$ for both
	$i \in \{1,2\}$.
	Then $h_1 \sim_A^D h_2$ and in particular
	$\sigma_a(h_1) = \sigma_a(h_2)$ for all $a \in A$.
\end{lemma}


\begin{proof}
We will actually prove the following properties by induction on $k$:
\begin{enumerate}
\item For every $a \in A$ and $v\in V_a$, $v \in h_1[k]$ iff $v \in h_2[k]$.

\item For every
$a,b\in A$, 
$Vis(h_1[k],a) = Vis(h_1[k],b) =  Vis(h_2[k],a)  =  Vis(h_2[k],b)$.
\item For each $c \not \in  A$, $a \in A$ and $v \in V_c  \cap Vis(h_1[k],a)$,
$v \in h_1[k] \text{ iff } v \in h_2[k]$.
\end{enumerate}
Note that all these properties imply that $h_1[\leq  k] \sim_A^D h_2[\leq  k]$.

The base case $k=0$ is trivial by the assumption that $h_1[1]=h_2[1]$,
so let's assume the two properties hold for some $k$. Since, as noted,
this implies that $h_1[\leq  k] \sim_A^D h_2[\leq  k]$, uniformity of
$\sigma_A$ implies that, $\sigma_A(h_1[k]) = \sigma_A(h_2[k])$.  But
then the effect of this unique tuple of actions on each atom
$v \in \bigcup_{a\in A} V_a$ is the same (actions are deterministic),
which ensures property (1).
 

To prove property (2), assume there exist two joint actions
$\gamma_{\overline{A}}^1, \gamma_{\overline{A}}^2 \in \prod_{c \not\in A} Act_c$
such that for each $i=1,2$, $h_i[k] \xrightarrow{\sigma_A(h_i[\leq
k]),\gamma_{\overline{A}}^i} h_i[k + 1]$.  Fix some $c \not \in A$ and denote
$\gamma_c^i = (\gamma_{\overline{A}}^i)_c$ for both $i=1,2$.  Also take some
$v \in V_c$ and two agents $a,b\in A$.  We will give the proof for
$h_1[\leq k + 1]  \sim_A^E  h_2[\leq  k+ 1]$, the general case following by
induction on the length of the path of indistinguishabilities
relating $h_1[\leq k + 1]  \sim_A^C  h_2[\leq k + 1]$.  And, for
this, we assume, without loss of generality, that $h_1[\leq k + 1]  \sim_a  h_2[\leq k + 1]$.

Assume now that $v  \in  Vis(h_1[k + 1],a)$.
The first assumption, together with $h_1[\leq  k + 1]  \sim_a  h_2[\leq  k + 1]$, implies that 
$v  \in  Vis(h_2[k + 1],a)$.
So, by condition ({\em Cmp}) in Def. \ref{restr}, for each $i=1,2$, $\gamma_c^i$ contains assignment $vis(v,a) := true$, 
and, by condition \ensuremath{(\dagger^{A})}, both $\gamma_c^i$ contain $vis(v,b) := true$.
But, in this way, we get that $v  \in  Vis(h_1[k + 1],b)  \cap  Vis(h_2[k + 1],b)$.
So $Vis(h_1[k + 1],a)  \subseteq  Vis(h_1[k + 1],b)  \cap  Vis(h_2[k + 1],b)$.
A similar argument holds if we start with $v  \in  Vis(h_1[k + 1],b)$ or $v  \not \in  Vis(h_1[k + 1],a)$
which proves that condition (2) holds.

For proving property (3), assume $v  \in  V_c  \cap  Vis(h_1[k + 1],a)$.  By
property (2) just proved, $Vis(h_1[k + 1],a)  =  Vis(h_1[k + 1],b)$ for every
$b \in  A$, hence $v  \in  Vis(h_1[k + 1],b)$ for every $b$.  By assumption
that $h_1[k + 1]  \sim_A^E  h_2[k + 1]$ and by definition of $\sim_A^E$,
there must exist some $b  \in  A$ such that $Vis(h_1[k + 1],b)  = 
Vis(h_2[k + 1],b)$, and for every $v'  \in  Vis(h_1[k + 1],b)$, $v'  \in 
h_1[k + 1]$ iff $v'  \in  h_2[k + 1]$. Hence this property must be met by
$v  \in  Vis(h_1[k + 1],b)$, which ends our proof.
\end{proof}

\begin{remark}
Lemma~\ref{lemma-two-histories} still
allows $A$-cast \CGS\ in which, for some histories $h_1$ and $h_2$,
$h_1 \sim_A^C h_2$ but $h_1 \not \sim_A^D h_2$ when $h_1[1]\neq h_2[1]$.  
That is, this lemma does not imply that common knowledge and distributed knowledge coincide on
histories in general.
\end{remark}

\underline{ {$A$-formulas}.} Let $A \subseteq Ag$ be a set of agents. We denote the  subclass of $ATL$ formulas where coalition operators only involve sets  $B \subseteq A$ of agents as
 \emph{$A$-formulas}. 
 
\begin{theorem}\label{thm-A-cast}
The model checking problem for $A$-cast \CGS{} and $A$-formulas is 2-EXPTIME-complete.
\end{theorem}

\begin{proof}
	The proof proceeds by structural induction on the given
	formula.  We illustrate the technique when the formula is
	either of type $\naww{A} p_1 U p_2$ or $\naww{A}p_1R p_2$,
	which can be easily generalized to the case of formulas
	$\naww{B} p_1 U p_2$ or $\naww{B}p_1R p_2$ when $B \subseteq
	A$, since then any $A$-cast \CGS\ is also a $B$-cast \CGS.
	The general case follows by the usual structural induction
	technique: for each subformula of the type $\naww{B} \phi$, we
	append a new atomic proposition $p_{\naww{B}\phi}$ which
	labels each state where coalition $B\subseteq A$ has a uniform
	strategy to achieve $\phi$.

	We construct a two-player game that generalises the
	information-set construction for solving two-player zero-sum
	games with imperfect information in \cite{vanderMeyden2005}.
	The difficult part is to show that the information sets do not
	``grow'' unboundedly.  To this end, it suffices that, for each
	history compatible with a joint $A$-strategy, we keep track
	just of the first and last state in the history.
        This is so
	since, by Lemma \ref{lemma-two-histories}, every uniform
	$A$-strategy will behave identically on two histories that
	share the first and last states and are related by the common
	knowledge relation for $A$.
	
	

More technically, we build a turn-based game $\hat \CGSname$ with
	perfect information in which the protagonist simulates each
	winning strategy for coalition $A$ on some finite abstractions of
	information sets.
Set $\hat S$ contains macro-states $V$ defined as tuples whose most
	important component are triples $(q,r,b)$, where $q \in S_0$
	and $r\in S$ represent the initial and the final state in a
	history, while $b$ is a bookkeeping bit for remembering
	whether the objective $p_1 U p_2$ has been reached or not
	during the encoded history.  Two triples $(q,r, b),(q',r',b')$
	in $V$ are connected via the common knowledge relation of
	$\CGSname$ extended to tuples. Hence, intuitively,
	macro-states along a run in $\hat \CGSname$
	abstract a set of histories of equal length and related by the
	common knowledge relation in $\CGSname$.

The essential step in the
	construction of $\hat \CGSname$ is to show that this
	abstraction is compatible with any uniform strategy. That is,
	for any two histories of equal length, starting in the same
	initial and ending in the same final state and related by the
	$\sim_A^C$ relation (on tuples of states), any uniform
	strategy prescribes the same action to any agent $a$ in
	coalition $A$. So, agent $a$ does not need to remember all
	intermediate steps for deciding which action to choose, only
	the initial state and the final state of each history suffice.
	By doing so, all strategies in the original game $\CGSname$
	are simulated by strategies in the resulting game
	$\hat \CGSname$, and vice-versa.

	We now formalise the construction of game $\hat \CGSname$.  To
	this end, given a set $ V \subseteq S \times
	S \times \{0,1\}$, we denote by $\approx_a$ the relation of
	$a$-indistinguishability extended to tuples: $(q,r,
	b) \approx_a (q',r', b)$ iff $q \sim_a q'$ and $r\sim_a r'$
	for any bit $b$. Then, $\approx_A^{C,V}$ is the common
	knowledge equivalence induced by $\{ \approx_a \}_{a\in A}$
	and restricted to $V$.  Since the bit $b$ does not play a role
	in the $\approx_a$-relation, we often simply write $(q,r) \approx_a
	(q',r')$.

	The set of states in $\hat \CGSname$ is $\hat S = \hat S_\Diamond\cup
	\hat S_\Box$, with $\hat S_\Diamond$ (resp.~$S_\Box$) denoting the protagonist (resp.~antagonist) states,
	and
	\begin{align*}
	\hat S_{\Diamond} & = \big\{  (V,\lambda) \mid V \subseteq S_0 \times S \times \{0,1\}, 
	\text{ for all } (q,r,b),(q',r',b') \in V, \\
	& \qquad (q,r) \approx_A^{C,V} (q',r') \text{ and } \lambda : V \to \{\emptyset, p_1,p_2\} \big\}  \\
	\hat S_{\Box} & = \big\{ (V,\lambda,\sigma) \mid (V,\lambda) \in S_\Diamond, \sigma : V \to Act_A \big\} \\
	\hat S_0 & = \big\{ (V,\lambda) \in \hat S_\Diamond \mid \text{ for all } (q,r,b) \in V, q =r , b=0 \text{ and }  \lambda(q,q,0) = q \cap \{p_1,p_2\} \big\}
	\end{align*}

	The game proceeds as follows: 
	in a macro-state $(V,\lambda)\in \hat S_{\Diamond}$ the protagonist 
	chooses a one-step strategy profile $\sigma : V \to Act_A$ 
	such that,
	for each $(q,r,b),(q',r',b') \in V$, if $r \sim_a r'$ for some $a \in A$, then $(\sigma(q,r,b))_a = (\sigma(q',r',b'))_a$.
	The control then passes to the antagonist which computes a new state as explained below:
	First, in the resulting state $(V,\lambda,\sigma)$, 
	one computes the set $W$ of outcomes:
	\begin{align*}
	W = \big\{ (q,s,b') \mid \text{ for some } (q,r,b) \in V, s = out(r,\sigma(r)) 
	\text{ and } b' = b \vee (p_2\in \lambda(s)) \big\}
	\end{align*}

Then, for each tuple $(q,r,b') \in W$, one computes
	$\widehat{W}(q,r,b') = \big\{(q',r',b'') \in W \mid
	(q',r') \allowbreak \approx_A^C (q,r)\big\}$.
	Additionally, a labelling $\lambda'$ is associated to $\hat W(q,r,b)$
	as follows: for each $(q,s,b) \in W$, let $\widehat{V}(q,s,b) = \big\{
	(q,r,b') \in V \mid s \in \tau(r,\sigma)\big\}$.  
	Then put:
	\begin{align*}
	\lambda'&(q,s,b) = 
	\begin{cases}
	\{p_2\} & \text{ if } p_2 \in s \text{ or } b=1 \\
	\{p_1\} & \text{ else if for all } (q,r,b') \in \widehat{V}(q,s,b), p_1 \in \lambda(q,r,\new{b'}) \\
	\emptyset & \text{ otherwise}
	\end{cases}
	\end{align*}
	The antagonist then picks a tuple $(q,r,b')\in W$ and designates the new protagonist state as the tuple $(\hat W(q,r,b'),\lambda'\restr{\hat W(q,r,b')})$.


	
	The objective of the protagonist is to reach macro-states $(V,\lambda)$ with 
	$\lambda(q,r,b) = p_2$ for all $(q,r,b) \in V$.  We can show that in
	$\CGSname$ coalition $A$ has a strategy to enforce $p_1 U p_2$ iff in
	$\hat \CGSname$ the protagonist has a winning strategy.  The easy part
	is the inverse implication, since each winning strategy for the
	protagonist in $\hat \CGSname$ gives enough information to construct a
	joint strategy for coalition $A$ in $\CGSname$ to achieve $p_1 U p_2$.
	As for the difficult part, consider a strategy $\sigma$ for coalition
	$A$ in $\CGSname$ achieving $p_1 U p_2$.  In order to build a strategy
	for the protagonist to win the game $\hat \CGSname$, we need to decide
	what action $f_A : V \to Act_A$ the protagonist chooses at a macro-set
	$(V,\lambda) \in \hat S_\Diamond$. Now, each triple $(q,r,b) \in V$ may represent 
	a whole family of histories
	$(h_l)_{l \in L}$ in $\CGSname$ with $h_l[1]=q$ and $h_l[|h_l|] = r$,
	all of them related by the common knowledge relation in \CGSname.
	So, one might face the problem 
	of consistently choosing one of the actions $\sigma(h_l)$ for playing at $(V,\lambda)$.  
	Fortunately, Lemma
	\ref{lemma-two-histories} ensures that whenever we have histories
	$h_1 \sim^C_A h_2$ with $|h_1| = |h_2|=n$, $h_1[1]=h_2[1]=q$, 
	$h_1[n]=h_2[n] = r$, and consistent with a uniform $A$-strategy $\sigma_A$, 
	then $\sigma_a(h_1) = \sigma_a(h_2)$ for all $a \in A$ -- and, as such, $\sigma(h_l)=\sigma(h_{l'})$ 
	for all $l,l' \in L$.
	Then, by denoting
	$\widehat V$ the history of macro-states that leads to $(V,\lambda)$ with $n$ denoting the length of this history, 
	the strategy
	$\hat \sigma$ can be chosen as follows: for each $(q,r,b) \in V$, choose \emph{any} history $h$ of length $n$ in $\CGSname$ 
	with $h[1] = \new{q}, h[n] = r$ which is compatible with $\sigma$ and put 
	$	\hat \sigma(\widehat V) (q,r,b) = \big(\sigma_a(h)\big)_{a \in A}$.
	It is not difficult to see that, whenever $\sigma$ achieves objective $p_1 U p_2$ in $\CGSname$, 
	$\hat \sigma$ is winning for the protagonist in $\hat \CGSname$.

	The case of $p_1 R p_2$ is treated similarly, with the reachability
	objective being replaced with a safety objective for avoiding
	macro-states $(V,\lambda)$ in which there exists $(q,r,b)$ with $\lambda(q,r,b) = \emptyset$.

	
	Finally, nesting of $A$-formulas is handled by 
	creating a new atom $p_{\phi}$ for each subformula $\phi$ of the given formula. 
	This atom is used to relabel the second component of each tuple $(q,r,b)$ belonging to each macro-state $(V,\lambda)$ as follows: 
	\begin{itemize}
	  \item If $\phi$ is a boolean combination of atoms, then append $p_\phi$ to $\lambda(q,r,b)$ iff $r\models \phi$.
	  \item If $\phi = \naww{B}p_1 U p_2$ for $B\subseteq A$, then append $p_\phi$ to $\lambda(q,r,b)$ iff 
	  in the two-player game constructed as above for coalition $B$
	  and starting in the set of initial states $(U,\lambda\restriction_{U})$ where $U = \{(q',r',b') \in V \mid \text{ for some } a \in B, (q',r') \sim_a^V (q,r)\}$,
	  the protagonist has a strategy for ensuring $p_1 U p_2$. A similar construction gives the labeling for $\phi = \naww{B} p_1  R p_2$. 	
	  \item If $\phi = \naww{B} X p$ for $B\subseteq A$, then append $p_\phi$ to $\lambda(q,r,b)$ iff
	  there exists a tuple of actions $f_B : U \sd Act_B$ (with $U$ defined at the previous item) such that 
	  all the macro-states $(U',\lambda')$ directly reachable by actions $f_A$ extending $f_B$ to the whole $Act_A$ 
	  have the property that for any $(q',r',b') \in U'$, $p \in r'$.
  \end{itemize}

To end the proof, we note that the above construction is exponential in size of the given \CGS.
Since we have a symbolic presentation of the \CGS, this 
gives the 2-EXPTIME upper bound. 
The 2-EXPTIME lower bound follows from the results on model-checking reactive modules with imperfect information \cite{GutierrezPW16}, p398, table 1, 
general case for two-player reactive modules, which can be encoded as a model-checking problem for singleton coalitions 
for \CGS.
\end{proof}
	



\section{Coalition-cast \CGS\ in Modelling Security Problems}
\label{example}

We now apply coalition-cast \CGS\ to the modelling of secure systems
and their threats. We start by recalling the setting of Terrorist-Fraud Attacks \cite{AD:2019}
and the distance-bounding protocol by Hancke and Kuhn \cite{HK05}, 
then give an $A$-cast \CGS\ model for this protocol and an ATL formula expressing the existence of an attack against the protocol,
and finally we argue why the results from Section 4 are relevant for this case study.

\paragraph{\textbf{Identity Schemes.}}
\label{subsec:hk}
 Distance-bounding (DB) protocols  are identity schemes which can be summarised as
  follows: (1) via timed exchanges, a  \emph{prover} $P$ demonstrates that it is physically situated within a bound
  from a  \emph{verifier} $V$; (2) via these
  exchanges, the prover also authenticates himself to the verifier.
Herein, we consider the distance-bounding (DB)  protocol by  Hancke and Kuhn~\cite{HK05}. A version of this protocol is summarised  in
Fig~\ref{hancke}.  In this protocol, a verifier $V$ will finish successfully if it authenticates a prover $P$ 
and if it is conceived that this prover is  physically situated no further than a given bound.
%
\begin{small}
 \begin{figure}[!h]
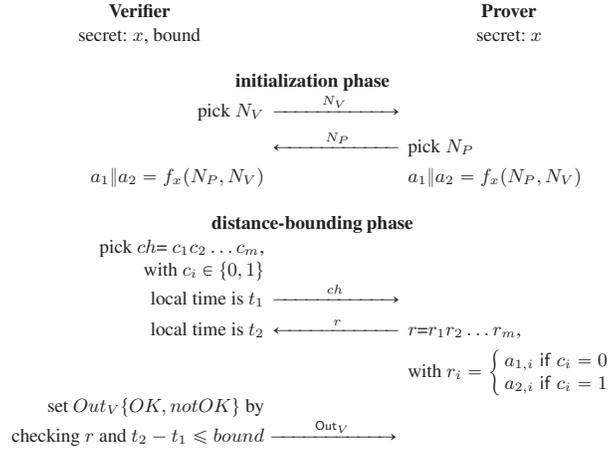

	\centering
	\begin{footnotesize}
		\scalebox{0.8}{
			\begin{tabular}{@{}rcl@{}}
				\multicolumn{1}{c}{\textbf{Verifier}} &&
				\multicolumn{1}{c}{\textbf{Prover}} \\
				\multicolumn{1}{c}{secret: $x$, bound} &&
				\multicolumn{1}{c}{secret: $x$} \\
				&& \\
				\multicolumn{3}{c}{\textbf{initialization phase}} \\
				pick $N_V$ & $\xrightarrow[\hspace{1.9cm}]{N_V}$ & \\
				& $\xleftarrow[\hspace{1.9cm}]{N_P}$ & pick $N_P$ \\
				$a_1\|a_2= f_x(N_P,N_V)$ && $a_1\|a_2= f_x(N_P,N_V)$ \\
				\phantom{$r_i=\left\{\begin{array}{ll}
					a_{2,i} & \mathsf{if\ }c_i=2
					\end{array}\right.$}
				&& \\
				\multicolumn{3}{c}{\textbf{distance-bounding phase}} \\
				pick $ch$= $c_1c_2 \ldots c_m$, \\ with $c_i\in\{0,1\}$ && \\
				local time is $t_1$ &
				$\xrightarrow[\hspace{1.9cm}]{ch}$ & \\
				local time is $t_2$  &
				$\xleftarrow[\hspace{1.9cm}]{r}$ &
				$r$=$r_1 r_2 \ldots r_m$, \\
				& & with $r_i=\left\{\begin{array}{ll}
				a_{1,i} & \mathsf{if\ }c_i=0 \\
				a_{2,i} & \mathsf{if\ }c_i=1 \\
				\end{array}\right.$ \\
				set $Out_V \{OK, not OK\}$ by \\
				checking  $r$ and $t_2 - t_1 \leq bound$ &
				$\xrightarrow[\hspace{1.9cm}] {\mathsf{Out}_V }$ \vspace*{-0.4cm}
			\end{tabular}
		}
	\end{footnotesize}
	\caption{A Depiction of the Hancke-Kuhn (HK) Protocol}
	\label{hancke}
\end{figure}
\end{small}
In this protocol, the prover and the verifier share a long-term key
$x$. In the initialisation phase,
$f_x(\ldots)$  produces a
pseudo-random bitstring $a_1 || a_2$,   where `` $| |$ '' is concatenation. In the distance-bounding phase, 
 $V $ pick randomly $m$ bits  $c_i$ and forms the challenge bitstring/vector $ch$. It sends this to the prover and times the time to it to respond.  
 $P$'s  response-vector $r$  is formed from bits $r_i$, where $r_i$ is the $i$-th bit of
$a_1$ if $c_i$ is $0$, or the $i$-th bit of $a_2$ if $c_i$ is $1$.
Finally, the verifier checks if the  response was correct and if it did not take too long to arrive. 

 \paragraph{\textbf{HK Versions.}}  In the original HK protocol~\cite{HK05}, the challenges and responses are sent not in one bitstring $ch$ and $r$ respectively, but rather bit by bit, i.e., $c_i$ and $r_i$. The one  timing step is also broken into $m$ separate timing. W.r.t. terrorist fraud (see description below),  it can be easily proven that if there is an attack on the version of protocol presented in Figure~\ref{hancke}, then there is an attack in the original protocol in~\cite{HK05} and vice-versa.  Formal symbolic verification methodologies~\cite{Blanchet2012} such as the recent~\cite{Debant18,Mauw18,AD:2019}  are only able to capture the HK version presented in Figure~\ref{hancke} (i.e., modelling 1-bitstring challenge and not $m$ 1-bit challenges); and this is true   w.r.t. checking  for all attacks/properties, not just terrorist frauds.   Unlike these, in our \CGS\, we are able to naturally encode both the original HK protocol and the depiction  in Figure~\ref{hancke}. Yet, in Section~\ref{subsec:example}, we present a \CGS\  encoding of  the  HK version in Figure~\ref{hancke}, as it yields a model  less tedious to follow step by step. If these verification passes, then $Out_V$=$1$, otherwise this $Out_V$=$0$.

 \paragraph{\textbf{Identity-fraud Attacks.}} A possible threat  is that of  a collusion-based attack called \emph{terrorist-fraud (TF)}. In this, a dishonest prover $P$, found
 far away from the honest verifier $V$, helps an adversary
$\mathcal{A}t$ pass the protocol as if $P$ was
close to $V$. However, a protocol is considered to suffer from such security problem, if  there is a strategy for 
a dishonest $P$ to give this help without jeopardising his authentication in future sessions. I.e., the help can aid $\mathcal{A}t$  in the collusive session alone, and thereafter $\mathcal{A}t$ can no longer authenticate as $P$.

It is known~\cite{HK05} that the HK protocol does suffer from a TF attack. Namely, a dishonest far-away
prover first colludes with the attacker in one execution, and gives him $a_1$ and $a_2$,
before the timed phase starts. After this data-sharing, the attacker can place himself close to $V$ and it
will be able to answer correctly and within the time-bound just during the execution  for which $a_1$ and $a_2$ were (privately) leaked by 
$P$ before the timed-phase started.

 \paragraph{\textbf{Attack finding on \CGS.}}
\label{subsec:example}  

We now give a coalition-cast \CGS\ that can be used to verify if, in the HK protocol, there exists a strategy by a (dishonest) prover $P$ and  the attacker  $\mathcal{A}t$  
 amounting to a TF attack (or the TF described above).

{\bf Modelling details}}.
The set of agents is $Ag = \{P, \mathcal{A}t, V\}$, i.e., prover $P$,
attacker $\mathcal{A}t$, and verifier $V$. The coalition in this $A$-cast \CGS\ is $A$ = $ \{P, \mathcal{A}t\}$.

%

We model {$n$
executions of the HK protocol under adversarial conditions. In
one arbitrarily-picked execution denoted $d$, the prover is far from $V$ and gives away his timed-phase response to the attacker.
Apart from such leakages, the prover $P$ behaves honestly as per the HK specification. The verifier $V$ is also honest.}

To check the existence of the TF attack in our \CGS\ modelling, it is sufficient to just model  the execution of
distance-bounding  phase and to encode  that
$P$ is close/far to $V$, without modelling the explicit distances/timing; this will be become clearer during the modelling we present, and we will 
re-explain at the end of said encoding. 

For simplicity, we use variables (as opposed to atoms). Boolean variables in the protocol equate to atoms; variables
of discrete and finite {enumerations} can be replaced by atoms, as per the \CGS\ definition. 
For the domain of the variables, our enumeration types are as follows:  $ch\_set$ for the domain of possible challenge-vectors $ch$,  
$enum^j_1$ for the domain of  prover-generated responses in each execution $j$, $enum_2$ for adversarially-generated responses. 
To capture the pseudo-randomness of  the responses, 
responses generated by  $P$ and $\mathcal{A}t$  will respectively range over  disjoint sets:
 each $enum^j_1$  used to model  execution $j$ of $P$'s is
disjoint from $enum_2$ used for $\mathcal{A}t$.  However, attacker $\mathcal{A}t$ will get some
values in $enum^k_1$ during the  collusion and its related data-sharing in an execution~$k$.

In all sessions, $\mathcal{A}t$'s aim is to forge $P$'s authentication. So, $\mathcal{A}t$
includes a Dolev-Yao (DY)
modelling~\cite{DY83}: \emph{$asg^{vals}_{DY}$} denotes that $\mathcal{A}t$
can update his variables based on their current values $vals$,
as per a DY attacker. The full DY specification is eluded due to
space constraints.

 Objects indexed by $j$ refer to HK's execution number $j \in \{1, 2,
\ldots, n\}$.

\textbf{Agents.} The set
\emph{$Var_P = \{ P\_far^j, P\_collude^j : \mathsf{bool}, P\_{r_{ch}}^j: \mathsf{enum^j_1}   \mid  j \leq n,  ch \in ch\_set \}$}  contains 
variables controlled\footnote{Variables controlled by a $\CGS$ agent $b$ are normally written $V_b$, but here we 
write $Var_b$ to avoid clashing with other notations in this protocols, i.e., the verifier.} by the prover $P$: whether $P$ is 
far from $V$ in the \mbox{$j$-th} HK-execution, whether $P$ is colluding in the \mbox{$j$-th} HK-execution, 
$P$'s 
response-vectors to each possible challenges $ch$ range over $enum^j_1$ in  the \mbox{$j$-th} HK-execution.

Similarly,  $\mathcal{A}t$'s  controlled variables are
\emph{$Var_{\mathcal{A}t} = \{ {\mathcal{A}t}\_{r_{ch}}^j: \mathsf{enum_2}  \allowbreak \cup  \mathsf{enum^j_1} \cup \{null\}  \, | \,  j \leq n, ch \in ch\_set  \}$},
 where $enum_2 \cap enum^j_1= \emptyset$ for each $j \in \{1, \ldots, n\}$.

The set $Var_V = \{ V\_c^j: ch\_set, V\_{r_{ch}}^j, V\_{rec_{ch}}^j: \mathsf{enum^j_1} \cup \mathsf{enum_2} \cup \{null\},\\
finished^j, ok^j : \mathsf{bool} \mid j \leq n, ch \in ch\_set \}$ contains  variables controlled by $V$. For each execution $j$, the verifier $V$ will actually calculated  a
  challenge-vector $V\_c^j$ as well as  a response-vector
$V\_{r_{ch}}^j$ calculated locally by $V$,  and it will receive response-vector $V\_{rec_{ch}}^j$ as the reply to a given value $ch$ that  is encapsulated inside $V\_c^j$. Finally, via the $finished$ and $ok$ variables $V$ will store if the executions have finished and
they were successful.
 
\textbf{Initial States.} For simplicity, we describe the initial states w.r.t.~each agent.
\begin{small}
\begin{eqnarray*}
 \textbf{\textbullet~initial states of the prover $P$:} & & \!P\_far^d,P\_collude^d
 \!:= \!\top; P\_collude^{k}\!:=\! \bot; 
P\_{r_{ch}}^j\!:=\! val^{j};  \\ \!\!\! \!\!\bigwedge_{u \in  \{P\_{r}^d\} } \!\!\!\!vis(u, \mathcal{A}t)\!: =\!\top
\end{eqnarray*}
\end{small}

\noindent
for each $j \in \{1, \ldots, n\}$,  for $d$ arbitrarily picked in $\{1, \ldots, n\}$,  for each $k \in \{1, \ldots, n\}$, 
$k\neq d$, for each $ch \in ch\_set$, with $val^{j}$  in
$enum^j_1$.

The above means that,  in the ``malign'' HK-execution $d$, $P$ is far-away
from the verifier and $P$ colludes with $\mathcal{A}t$ in this HK-execution. 
Notice that for rest of the executions (i.e., $j \neq d$), the
position of $P$ is left to be
assigned non-deterministically.  Also, $P$ pre-calculates the
responses for any possible challenge, in each HK-execution $j$. Plus,
$P$ sees all his variables. As per the collusion, $P$ makes the response
for the $d$-th HK-execution visible to $\mathcal{A}t$.


\begin{small}
\begin{tabbing}
$ \textbf{\textbullet~ initial states of the verifier $V$:} \; \; \; $ \= $     V\_c^j:=ch^j;  V\_{rec_{ch}}^j:= null;$
 $\bigwedge_{u' \in Var_V} \!vis(u',V): = \top;$ \\
$\bigwedge_{u \in \{V\_c^j | 1 \leq j \leq n \}, E \in \{ \mathcal{A}t,P\}} \!vis(u,E) := \top$,  $V\_{r_{ch}}^j:= val^{j};$ $finished^j, ok^j :=\bot;$  
\end{tabbing}
\end{small}
for each $j \in \{1, \ldots, n\}$, for each $ch \in ch\_set$. So, $V\_c^j:=ch^j$ equates to $V$ setting the challenges for
the \mbox{$j$-th} HK-execution; each $ch^j$ is chosen randomly in $ch\_set$. Also, $V$ has received no response (i.e., $V\_{rec_{ch}}^j:= null$), the HK-execution $j$ has not finished, nor was it checked (i.e., $finished^j, ok^j :=\bot$). And, $V$ sees his
variables (i.e.,  $\bigwedge_{u' \in Var_V} \!vis(u',V): = \top$), releases all challenges for each HK-execution~$j$ (i.e., $\bigwedge_{u \in \{V\_c^j | 1 \leq j \leq n \}, E \in \{ \mathcal{A}t,P\}} \!vis(u,E) := \top$), and  has  computed his local 
responses with values $val^{j}$ (i.e.,  $V\_{r_{ch}}^j:= val^{j}$). 

The values {$val^j$} inside the initial states of $P$ and $V$ are produced via a  
function of $j$, a constant $x$, and the set $ch\_set$, with outputs uniformly in $enum^j_1$, which emulates the pseudorandom function $f_x$ in the HK protocol.
As the HK attacker does not know $x$,    $\mathcal{A}t$ can
only generate his responses over another domain,  $enum_2$.
%

\begin{small}
	$ \textbf{\textbullet~initial states of the attacker ${\mathcal{A}t}$:} {\mathcal{A}t}\_{r_{ch}}^j: =null; \bigwedge_{u \in Var_{ {\mathcal{A}t}}} vis(u, \mathcal{A}t):= \top; \\ \bigwedge_{E \in \{P,V\}} vis({\mathcal{A}t}\_{r_{ch}}^j,E): =\bot, $
\end{small}
\noindent
for each $j \in \{1, \ldots, n\}$, for each $ch \in ch\_set$.

So,   $\mathcal{A}t$ has not yet computed his responses ${\mathcal{A}t}\_{r_{ch}}^j$ (i.e., ${\mathcal{A}t}\_{r_{ch}}^j: =null$).  Then, $\mathcal{A}t$ sees all his variables (i.e.,  $\bigwedge_{u \in Var_{ {\mathcal{A}t}}} vis(u, \mathcal{A}t):= \top$), but $\mathcal{A}t$'s responses 
are not yet visible to $V$ or to $P$ (i.e., $\bigwedge_{E \in \{P,V\}} vis({\mathcal{A}t}\_{r_{ch}}^j,E): =\bot$ ).



\textbf{The $\texttt{update}$ commands for the prover $P$.} 
The commands  $\mathbf{\gamma^{k,ch}_1}$ capture that, in any  honest HK-execution $k$, $P$ is ready for any challenge $ch$ (i.e., $val(V\_c^k \text{ in } vis(V\_c^k,P)) =ch$), and it will
answer with  $P\_{r_{ch}}^k$. Both $V$ and $\mathcal{A}t$ see the answer, as the setting inside $vis(\cdot)$  below shows:
\begin{small}
\begin{tabbing}

	$ \mathbf{\gamma^{k,ch}_1}  \; \; \neg {P\_far}^{k} \!\!\!\wedge\!\! val(V\_c^k\! \text{ in } vis(V\_c^k\!\!,P)) \!=\!ch \!\rightsquigarrow$ \=$ \!    \bigwedge_{E \in \{ \mathcal{A}t,V\}} \!\!vis(P\_{r_{ch}}^k,E):=\top$	
\end{tabbing}
\end{small}
for all $k$, $1\leq k \leq n$, $k \neq d$ and for each  $ch \in ch\_set$.

\textbf{The $\texttt{update}$ commands for $\mathcal{A}t$.} 
Via the commands $\mathbf{\gamma^{ch}_2}$, $\mathcal{A}t$ copies locally the response-table that $P$'s  gives him in the dishonest execution $d$.
\begin{small}
\begin{tabbing}
	$ \mathbf{\gamma^{val_{ch}^d}_2} \; \; \; val(P\_{r_{ch}}^d \text{ in } vis(P\_{r_{ch}}^d, \mathcal{A}t))= val_{ch}^d
  \rightsquigarrow$ \= $ {\mathcal{A}t}\_{r_{ch}}^d: =val_{ch}^d$
\end{tabbing}
\end{small}
for each value $val_{ch}^d \in enum^d_1$ set in the
initial states of $P$.

We add commands $\gamma^{vals}_3$ below:
\begin{small}
\begin{tabbing}
	$ \mathbf{\gamma^{k,vals}_3}$ \;   ``$Var_{\mathcal{A}t}$ take values $vals$''  $\!\rightsquigarrow\!  asg^{vals' \subset enum_3}_{DY}( {\mathcal{A}t}^k\_{r_{ch}})$ 
\end{tabbing}
\end{small}
for all $k \leq n$, $k \neq d$, for all $ch \in ch\_set$.
%

The last command above denotes  the Dolev-Yao capabilities of our attacker. Namely,
for each $k$-th HK-execution  in which $\mathcal{A}t$ is not helped by $P$, for all values
 $vals$ at a current state of $\mathcal{A}t$'s, then $\mathcal{A}t$  composes possible responses.  
 Each such composed value lies in
 a set $enum_3$, which includes $\mathcal{A}t$'s ``own'' values in $enum_2$,
 responses that $P$ collusively gave to $\mathcal{A}t$ (i.e., $\{{P}\_{r_{ch}}^d$\}) in $enum^{d}_1$, together
with values in  $enum^{k}_1$ that $\mathcal{A}t$  collects when $P$ sends responses out  via $ \mathbf{\gamma^{k,ch}_1}$  commands.
 I.e.,
 $enum_3$= $enum_2 \cup \{ {P}\_{r_{ch}}^d  | ch \in ch\_set\} \cup \bigcup_{k \leq n, k \neq d} \{
 P_{r_{ch}}^k \, \mid \, V\_c^k =ch \wedge P\_collude^k = \bot | \text{ for some } ch \in ch\_ set \}$.
 
%
%

The commands $\gamma^{j,ch}_4$, for each  $j \leq n$ and each $ch \in ch\_set$, say
that $\mathcal{A}t$ answers $V$'s challenge in each HK-execution (that is, $ \mathcal{A}t$ set the $vis(\cdot)$ atoms for $V$ accordingly, as per the below):
\begin{small}
\begin{tabbing}
	$ \mathbf{\gamma^{j,ch}_4}$ \; $val(V\_c^j \!\text{ in } vis(V\_c^j, \mathcal{A}t)) \!=\!ch \wedge {\mathcal{A}t}\_{r_{ch}}^j \neq null   \!\rightsquigarrow\!  vis({\mathcal{A}t}\_{r_{ch}}^j,V):=\top$
\end{tabbing}
\end{small}

\textbf{The $\texttt{update}$ commands for $V$.} 
%
Then, in each execution $j, 1 \leq j \leq n$,  for any  challenge $ch \in ch\_set$ it may have sent, $V$ will store the response  from  $\mathcal{A}t$ and/or $P$, via commands $\mathbf{\gamma^{j,ch}_5}$ and $\mathbf{\gamma^{j,ch}_{6}}$ (this storing is done via e.g. $V\_{rec_{ch}}^j: = val({\mathcal{A}t}\_{r_{ch}}^j)$):
\begin{small}
\begin{tabbing}
	$ \mathbf{\gamma^{j,ch}_{5}} $ \=  \; \; \;  \= $  (vis({\mathcal{A}t}\_{r_{ch}}^j, V)  \wedge \neg finished^j) \rightsquigarrow  V\_{rec_{ch}}^j: = val({\mathcal{A}t}\_{r_{ch}}^j)$\\
	$ \mathbf{\gamma^{j,ch}_{6}} $ \=  \; \; \; \= $  (vis(P\_{r_{ch}}^j, V)  \wedge \neg finished^j) \rightsquigarrow  V\_{rec_{ch}}^j: = val(P\_{r_{ch}}^j)$
\end{tabbing}
\end{small}

So, as with all commands of this type, note that $ \mathbf{\gamma^{j,ch}_{5}} $  is in fact a class of commands, for any possible value of $ch$, for any possible value/id of an execution $j$.

 Finally, $V$ can check the response, in each execution $j \leq n$, whichever the challenge 
 $ch \in ch\_set$ was, and the $ok^j$ are set if it authenticates $P$:
\begin{small}
\begin{tabbing}
	$ \mathbf{\gamma^{j,ch}_{7}} $ \=  $ (V\_c^j =ch \wedge  V\_{rec_{ch}}^j =V\_{r_{ch}}^j \wedge  \neg finished^j) \rightsquigarrow $  $finished^j :=\top; ok^j:=\top $ \\
	$ \mathbf{\gamma^{j,ch}_{8}} $ \= $ (V\_c^j \!=\!ch \!\!\wedge\!\!  V\_{rec_{ch}}^j \!\neq\! null  \!\!\wedge\!\! V\_{rec_{ch}}^j \!\neq\! V\_{r_{ch}}^j \!\!\wedge\!\!  \neg finished^j) \!\rightsquigarrow\!  finished^j \!:=\!\top; ok^j\!:=\!\bot $
\end{tabbing}
\end{small}

Lastly, we will like to note that in each of the  commands $ \mathbf{\gamma^{j,ch}_{5}} $, $ \mathbf{\gamma^{j,ch}_{6}} $,  $ \mathbf{\gamma^{j,ch}_{7}} $, $ \mathbf{\gamma^{j,ch}_{8}} $, there is an additional update, namely $\bigwedge_{ E \in \{ \mathcal{A}t,P\}}vis(\cdot, E): = val(vis(\cdot, E))$. This is to say that the verifier $V$, at each update commands, set the visibility atoms for the prover and the adversary; in this model, the verifier $V$ sets the values always to what they were beforehand (which is what  $\bigwedge_{ E \in \{ \mathcal{A}t,P\}}vis(\cdot, E): = val(vis(\cdot, E))$ denotes). This update is added to comply with condition $[{\rm (}Cmp{\rm)}]$ of Def.~\ref{restr} of $A$-cast systems; we did not write it in the commands above for sake of readability.

~

\textbf{The TF-encoding Formula.}  
We use the following self-explained  propositions:
\begin{itemize}
\item $\mathsf{auth\_via\_help}$  as: $(P\_far^{d} \wedge P\_collude^{d}) \to  ok^d$
\item $\mathsf{noAuth\_after\_help}$ as:
\begin{small} 
$\bigwedge_{j \in \{d, \ldots,  n\}, j \neq d}\big(  (P\_far^{d} \wedge P\_collude^{d})  \rightarrow (  P\_far^{j}  \rightarrow \neg ok^{j}) \big)$
\end{small}
\end{itemize}


Then, our TF-encoding formula says: there exists a
collusion between the far-away prover $P$ and attacker $\mathcal{A}t$
such that execution $d$ ends successfully, yet whatever 
$\mathcal{A}t$ does,
she cannot pass in a subsequent  execution $j$ as a simple man-in-the-middle:
\[\naww{P,\mathcal{A}t} F ( \mathsf{auth\_via\_help} \Rightarrow \all{\mathcal{A}t} G\,  \mathsf{noAuth\_after\_help}).\]

~

\textbf{Relevance of results in Section~\ref{decproof}.} 
We outline here the main arguments which support the relevance of the results for the verification of the existence of a terrorist-fraud attack 
in our case study. 

Firstly, note that in the TF model, Attacker and Prover need to share private information (Prover's response-vector) without Verifier immediately receiving it.
As argued in the introduction, the setting of \CGS\ offers a more elegant framework for 
sharing private information, even though any \CGS\ can be encoded as an interpreted system \cite{LomuscioQuRaimondi15},
as shown in \cite{BBDM19}, at the expense of doubling the number of variables.

Secondly, note that the existence of a TF attack utilises a non-singleton coalition operator.
In the presence of a memoryful interpretation of ATL with imperfect information -- which is consistent with the 
Dolev-Yao hypothesis of attacker capabilities -- abd in the absence of Section~\ref{decproof} the only possible approach 
would be to rely on the decidability results for ATL with imperfect information with the state-based (or memoryless) 
interpretation which are implemented in the MCMAS tool,
plus the fact that the model has a finite number of runs, all eventually looping in a final state.
To take this approach, one would need to explicitly encode agent memory in the model, that is,
for each variable $v$ in the model, create $N$ copies of it, where $N$ is the length of the longest run in the system (before ending in its looping state). 
By observing that the TF model is $A$-cast, we are then able to apply the model-checking algorithm provided by Theorem 1.
Certainly, this algorithm builds a memoryful strategy for the coalition formed by Attacker and Prover, 
but the memory needed by these two agents (and encoded in their winning strategy)
might be much smaller than what would be produced by explicit encoding of each agent memory in the model.

Finally, we note that expressing collusion requires a feature which is present in $ATL$:
nesting of coalition operators. This feature is not available in either Distributed Synthesis or Dynamic Epistemic Logic.
Also note that our TF-encoding formula is meaningful in the subjective interpretation of $ATL$,
due to the fact that the winning strategy for the coalition composed of Attacker and Prover
has to be uniformly applicable to all initial states, independently of the 
initialisation of the bitstring $a_1 || a_2$ which is not known initially by the Attacker.



%


\section{Related Work and Conclusions}
\label{sec:relwork}
%
%
The work here presented is related to contributions in distributed and multi-agent systems, 
including
propositional control, reactive modules, broadcast, dynamic epistemic
logic. However, it differs from the state of the art in numerous, subtle
aspects.

The \textbf{coalition logic of propositional control} (CL-PC)
was introduced in
\cite{CLPC2015}
as a logic to reason about the capabilities of agents in multi-agent
systems. It has since been investigated in relation with the transfer
of control \cite{HoekWW10}, the dynamic logic of propositional
assignments \cite{GrossiLS15}, non-exclusive
control \cite{Gerbrandy06}.
Further, in
\cite{BelardinelliGHLLNP17} it is proved that shared control
can be simulated by means of exclusive control.
However, the present contribution
differs from most of the works in this line in two key
aspects. Firstly, the references above deal mainly
with \emph{coalition logic}, which is the ``next'' $\naww{A} X$
fragment of $ATL$. Secondly, they all assume perfect information on
the agents in the system, while here we analyse the arguably more
challenging case of imperfect information.

\textbf{Visibility atoms} have been used to model
explicitly individual knowledge \cite{GLNP17}, even though only
in \cite{BBDM19} were used in the specific version considered in this
paper.
A notion of (propositional) visibility is analysed in
\cite{Maffre16} and the works cited therein.  In particular, similarly
to what we do here, the indistinguishability relation between
epistemic states is defined in terms
of the truth values of observed atoms. Still, a key difference
with \cite{Maffre16} and related works is that they introduce modal
operators for visibility as well. As a result, their language is
richer as it can express introspection for instance. On the other
hand, logics for strategies are not considered in this line of
research, while they are the focus of this paper.

Guarded commands have appeared in the framework of {\bf Reactive
Modules} \cite{AlurH99b}, which \CGS\ are reminiscent of. Model
checking $ATL$ on reactive modules has been investigated
in \cite{HoekLomuscioWooldridge06}, but in a perfect information
setting only.  More recently, imperfect information in reactive
modules has been explored in \cite{GutierrezPW16,GutierrezHW17},
however mainly in relation with the verification of game-theoretic
equilibria,
while here we study the verification of full $ATL$ under imperfect
information.

The semantics that we use also compares to the
event models of {\bf Dynamic Epistemic Logic}
\cite{hvdetal.del:2007}.  The underlying philosophy of event models is
to provide a flexible formalism for representing one-step system
evolutions in aspects regarding both the way the system's state and
information available to each agent change. However, event models do
not focus on the local implementation of actions in the way iCGS
do, and therefore they only accommodate a limited type of ``strategic''
reasoning.
We uphold that the combination of agent-based action specifications
(commands) and visibility atoms allows us to reason
about both agent strategies and information updates.

Encoding {\bf security properties} in coalition-expressing logics has
 very seldom been done. Using $ATL$ without nested modalities,
 properties pertaining to contract-signing \cite{Chadha:2006}, as well
 as
receipt-freeness and coercion-resistance in
 e-voting~\cite{TabatabaeiJR16,BelardinelliCDJ17,Selene18} have been
 encoded.  By contrast, in terms of the property expressed here
 (i.e., terrorist-fraud attacks/resistance in identification
 protocols), the nesting of strategic abilities is required in the
 specification.
Also, most of these prior works
 do not operate
 in a perfect recall semantics, and the coalitions expressed do not
 necessarily pose on private-information sharing, as do ours in the
 collusion setting.
 
{\bf Symbolic security verification of distance bounding (DB)} as opposed to computational analysis~\cite{Blanchet2012} has
 emerged only in the last two years~\cite{Mauw18,Debant18,AD:2019}. As the most recent line of this type~\cite{AD:2019} shows, 
 TF attacks can  be
expressed therein only as a simple approximation by imposing sequence of events. That is because in the process-algebraic~\cite{Blanchet01} or rewriting-based approaches~\cite{tamarin} used therein, collusions cannot be expressed  either the system-modelling level or at the property-encoding level.  
Nonetheless,
they do model security more faithfully than we do in Section~\ref{example}, but this is primarily due to the fact that herein we present our model in a pen-and-paper version.
As such, we believe our model to be the first formalism to faithfully lend itself to the analysis of terrorist frauds in secure system.

%

Finally, the restriction to {\bf broadcast} as mode of communication
to retain decidability of the model checking problem has been 
explored in \cite{BLMR17,BelardinelliLMR17b,DBLP:conf/lics/KupfermanV01,vanderMeyden2005}.
We have compared these approaches to ours in the beginning of Section 4.
More generally, other restrictions on iCGS have been investigated
recently, notably the hierarchical systems
in \cite{BerthonMM17,BerthonMMRV17b}, which are nonetheless orthogonal
to the present setting. In particular, notice that no hierachy is
assumed on the observations of agents in \CGS, nor anything similar
appears in standard TF attacks.

As \textbf{future work}, we plan to implement the model-checking algorithm in the MCMAS tool \cite{LomuscioQuRaimondi15}.
To date, this tool does not support communication between agents by means of shared variables, the only communication mechanism being synchronisation on joint actions. 
Therefore, even though our case study can be modelled and verified in MCMAS due to the fact that the model contains a finite number of lasso-type runs, 
the need to encode shared variables as synchronisation on joint actions would increase even more the complexity of the models, fact which, combined 
with the necessity to encode agent memory into the model in order to cope with the  ``memoryless'' semantics of ATL,  would have a non-negligible impact on the performances of the tool.
Therefore, our future research plans imply (1) enriching the syntax of the ISPL to allow for shared variables 
\emph{and} dynamic visibility of state variables through the addition of \emph{vis} statements, and (2) 
testing a model for the $A$-cast property, designing and implementing a symbolic version the model-checking algorithm in Section 5. 
Note however that ATL with imperfect information does not allow for fixpoint expansions of all coalition modalities, as it is the case with ATL with perfect information.

%






\bibliographystyle{splncs04}
\bibliography{bibliography}

\end{document}